\documentclass[graybox]{svmult}
\usepackage{mathtools}
\usepackage{array}
\usepackage{graphicx,color,bm,amsmath}
\usepackage[figuresright]{rotating}
\usepackage[]{graphicx}
\graphicspath{{../Figures/PDF/}{../Figures/Sources/PNG/}}

\usepackage{amssymb}
\usepackage{lineno}

\usepackage{float}
\usepackage{multirow}
\usepackage{booktabs}
\usepackage{subcaption}            % L\'egendes pour sous- float
\usepackage{array}
\usepackage{xspace}   				% Pour avoir des espaces malins dans les formules
\usepackage{mathrsfs}
\usepackage{mdwlist}                 % Pour des listes sans trop d'espaces.

\usepackage{subcaption}
\usepackage[english]{babel}
\usepackage[cyr]{aeguill}

\usepackage[linesnumbered,lined,boxed,commentsnumbered]{algorithm2e}

\usepackage[bordercolor=black,backgroundcolor=orange!60!white,linecolor=orange!60!white,figwidth=0.8\textwidth
]{todonotes}

\usepackage{listings}
\usepackage{xcolor} % for setting colors

\usepackage[colorlinks=true,linkcolor=blue!20!black,urlcolor=black,citecolor=blue!20!black,hyperfootnotes=false]{hy
  perref}

\usepackage[normalem]{ulem}

\usepackage{dsfont} 
\newcommand{\reals}{\ensuremath{\mathds{R}}}

\newcommand{\Rset}{{\mathbb R}}
\newcommand{\At}{{\mathbb A}}

\newcommand{\vt}{{\mathbf d}}
\newcommand{\vd}{{\mathbf d}}

\newcommand{\ctr}[2]{\ensuremath{[\vd^{#1},\vd^{#2}]}}

\newcommand{\shspace}{\mathcal{H}_4}
\newcommand{\cspace}{\mathcal{F}}
\newcommand{\tensspace}{\mathcal{A}}

\begin{document}

%\begin{frontmatter}



%%%\dochead{25th International Meshing Roundtable}
%% e.g. \dochead{17th International Conference on Dynamical Processes in Excited States of Solids}

%-----------------------------------------------------------------------------------------------------------------------------
%-----------------------------------------------------------------------------------------------------------------------------

\title*{Representing three-dimensional cross fields using 4th order tensors}


\institute{Alexandre Chemin \at Universit\'e catholique de Louvain, Avenue Georges Lemaitre 4, bte L4.05.02, 1348 Louvain-la-Neuve, Belgium, \email{Alexandre.Chemin@uclouvain.be} \and Fran{\c c}ois Henrotte \at Universit\'e catholique de Louvain, Avenue Georges Lemaitre 4, bte L4.05.02, 1348 Louvain-la-Neuve, Belgium, \email{francois.henrotte@uclouvain.be}
  \and Jean-Fran{\c c}ois Remacle \at Universit\'e catholique de
  Louvain, Avenue Georges  Lemaitre 4, bte L4.05.02, 1348
  Louvain-la-Neuve, Belgium, \email{jean-francois.remacle@uclouvain.be}
  \and Jean Van Schaftingen \at Universit\'e catholique de
  Louvain, Chemin du Cyclotron 2 bte L7.01.02 - 1348 Louvain-la-Neuve - Belgique \email{jean.vanschaftingen@uclouvain.be}}

\authorrunning{A. Chemin et al.} 
\author{Alexandre Chemin, Fran{\c c}ois Henrotte, Jean-Fran{\c c}ois
  Remacle and Jean Van Schaftingen}

%  \address[a]{Universit\'e Catholique de Louvain, MEMA, Avenue Georges Lemaitre 4, bte L4.05.02, 1348 
%    Louvain-la-Neuve, Belgium}
%  \address[b]{Universit\'e Catholique de Louvain, Mathematical Institute, Belgium}
%\address[b]{Second affiliation, Address, City and Postcode, Country}
\maketitle
\abstract{
  This paper presents a new way of describing cross fields based on
  fourth order tensors. We prove that the new formulation is forming a
  linear space in $\reals^9$. The algebraic structure of the tensors and
  their projections on $\mbox{SO}(3)$ are presented. The relationship of
  the new formulation with spherical harmonics is exposed. This paper is
  quite theoretical. Due to pages limitation, few practical aspects
  related to the computations of cross fields are exposed. Nevetheless,
  a global smoothing algorithm is briefly presented and computation of
  cross fields are finally depicted. }

%  \begin{keyword}
%    Fourth order tensors \sep Cross Fields \sep Hex-meshing




%  \end{keyword}
%  \cortext[cor1]{Corresponding author. Tel.: +32 10472355; fax: +32 10472999.}

%\cortext[cor1]{Corresponding author. Tel.: +32 10472355; fax: +32 10472999.}
%\end{frontmatter}
%\email{Jean-Francois.Remacle@uclouvain.be}

%\correspondingauthor[*]{Corresponding author. Tel.: +0-000-000-0000 ; fax: +0-000-000-0000.}

%\linenumbers
\section{Introduction}

We call a cross $f$
a set of  $6$ distinct unit vectors
mutually orthogonal or opposite to each other (Fig.~\ref{fig::crosses}). 
This geometric object of vectorial nature
lives in the tangent space of Euclidean spaces $E^3$. 
A cross field $F = \{x \in \Omega \subset E^3 \mapsto f(x)\}$, now, 
is a rule that associates a cross $f(x)$ to each point of a subset $\Omega$ of $E^3$.
Cross fields are auxiliary in 3D mesh generation 
to define local preferred orientations for hexahedral meshes,
or for the computation of the polycube decomposition of a solid. 
Automatic polycube decomposition is a necessary step for
multiblock or isogeometric meshing of 3D domains.

Let the Euclidean space $E^3$ be equipped 
with a Cartesian coordinate system $\{x_1, x_2, x_3\}$.
The six vectors $(\pm 1, 0,0)$, $(0, \pm 1, 0)$ and $(0, 0, \pm 1)$
form a cross, which we call the reference cross $f_{ref}$.
Crosses being rigid objects, their orientation in space
can be \emph{identified} by a rotation respective to $f_{ref}$,
that is a member of $\mbox{SO}(3)$
represented by, e.g.,  the Euler angles $\alpha$, $\beta$ and $\gamma$, 
(Fig.~\ref{fig::crosses}). 
This \emph{representation} of $f$ is however not unique
due to the symmetries of the cross,
which are fully characterized by regarding the cross
as set of six points at the summits of an octahedron.
The symmetry group of this point set has 24 elements,
which are the 24 rotations that apply the cross onto itself,
and is called the octohedral point group $\mbox{O}$.
We call attitude of the cross $f$
its orientation in space up to the symmetries of the cross,
and we have 
$$
f \in \mbox{SO}(3) / \mbox{O}.
$$

\begin{figure}[!ht]
  \begin{center}
    \includegraphics[width=0.3\textwidth]{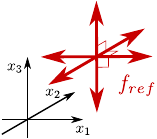}
    \includegraphics[width=0.3\textwidth]{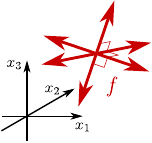}
  \end{center}
  \caption{3D crosses representation. Left image shows the reference
    cross $f_{ref}$ and right image shows a cross $f$ that is a
    rotation of the reference cross.}
  \label{fig::crosses}
\end{figure}

The expansion of a discretized field $F$ into coefficients and shape functions
in finite element analysis is by definition a linear combination,
leading then by orthogonalization in this linear space 
to a linear system of equations to solve.
It is hence necessary in this finite element context to have 
a representation of the point value of the discretized field $F(x)$
in a linear space,
i.e., a space containing the linear combinations
(here with real coefficients) of all its members. 
This is however in general not the case 
with fields taking their values in non-trivial group manifolds
like, e.g., $\mbox{SO}(3) / \mbox{O}$.

The sketch of the solution to this problem 
can be illustrated with a simple 2D example. 
Consider the unit circle $S^1$ 
and two points $e^{\imath \theta_1}$ and $e^{\imath \theta_2}$ on this manifold.
Clearly, linear combinations 
$$
a  \,  e^{\imath \theta_1} + b  \, e^{\imath \theta_2}
\quad , \quad 
a,b \in \reals
$$
do not all belong to $S^1$.
In order to have a practical representation of the elements of $S^1$
in a linear space amenable to finite element analysis,
one has to expand $S^1$ to the enclosing complex plane, 
$\ensuremath{\mathds{C}} \supset S^1$,
which is a linear space.
The finite element problem can so be formulated in terms of complex valued unknowns
that are afterwards projected back into $S^1$
by means of a projection operator, e.g., 
$$
\Pi\,: \ensuremath{\mathds{C}} \mapsto S^1
\quad , \quad 
x+ \imath y \mapsto e^{\imath\,\mbox{atan2}(y,x)}.
$$

A similar approach is followed in this paper 
for the 3D finite element smoothing of cross attitudes
belonging to the group manifold $\mbox{SO}(3) / \mbox{O}$.
The approach rely on a new way of representing 3D cross fields
as a particular class of $4^{th}$ order tensors,
themselves in close relations to $4^{th}$ degree homogeneous polynomials
of the Cartesian coordinates. 
3D cross field representations based on tensors have been used 
for 3D solid texturing and hex-dominant meshing
\cite{takayama2008lapped,zhang2011sketch,vyas2009tensor,yamakawa2003fully},
but none of them was adressing symmetry issues or projections. 
The use of $4^{th}$ order tensors allows to build 
a $9$-dimensional linear space $\tensspace$, 
containing $\mbox{SO}(3) / \mbox{O}$ as a subset,
together  with a projection operator 
$$
\Pi\,: \tensspace  \mapsto \mbox{SO}(3) / \mbox{O}.
$$ 
The approach leads eventually to a very efficient smoother for cross fields,
one order of magnitude faster than state-of-the art implementations. 
The proposed representation also allows easy computation
of the distance between a finite element computed cross $f$,
and its projection back into $\mbox{SO}(3) / \mbox{O}$.
This distance indicates the presence of singular lines and singular points 
in the cross field in a straightforward fashion.  

The paper is organized as follows.
The $4^{th}$ order tensor representation for crosses
is first inroduced,
and the useful mathematical properties of this tensor space are then derived. 
The projection method is then presented 
and results obtained with a naive 3D
crossfield smoothing on some benchmarks problem
are finally discussed.

\section{Cross representation with $4^{th}$ order tensors}

\subsection{The reference cross $f_{ref}$}

Point groups, like $O_h$, are isometries leaving 
at least one point of space, the center, invariant. 
As such, they have very convenient and useful representations on the sphere,
and hence also in terms of spherical harmonics.
In \cite{huang2011boundary,ray2016practical}, 
spherical harmonics of degree $4$
are proposed as a polynomial basis to represent 3D cross fields.
They exhibit the required octahedral symmetry
and span a linear polynomial space $\shspace$ of dimension $9$.
The projection operator 
$$
\Pi\,: \shspace \mapsto \mbox{SO}(3) / \mbox{O},
$$
however, is tedious as it relies on a complex minimization process 
that is not ensured to converge to the true projection.
Moreover, the differential properties of spherical harmonics
(they are solution of the laplacian operator)
are of no use to the purpose of cross representation. 

\begin{figure}[H]
  \begin{center}
    \includegraphics[width=.4\linewidth]{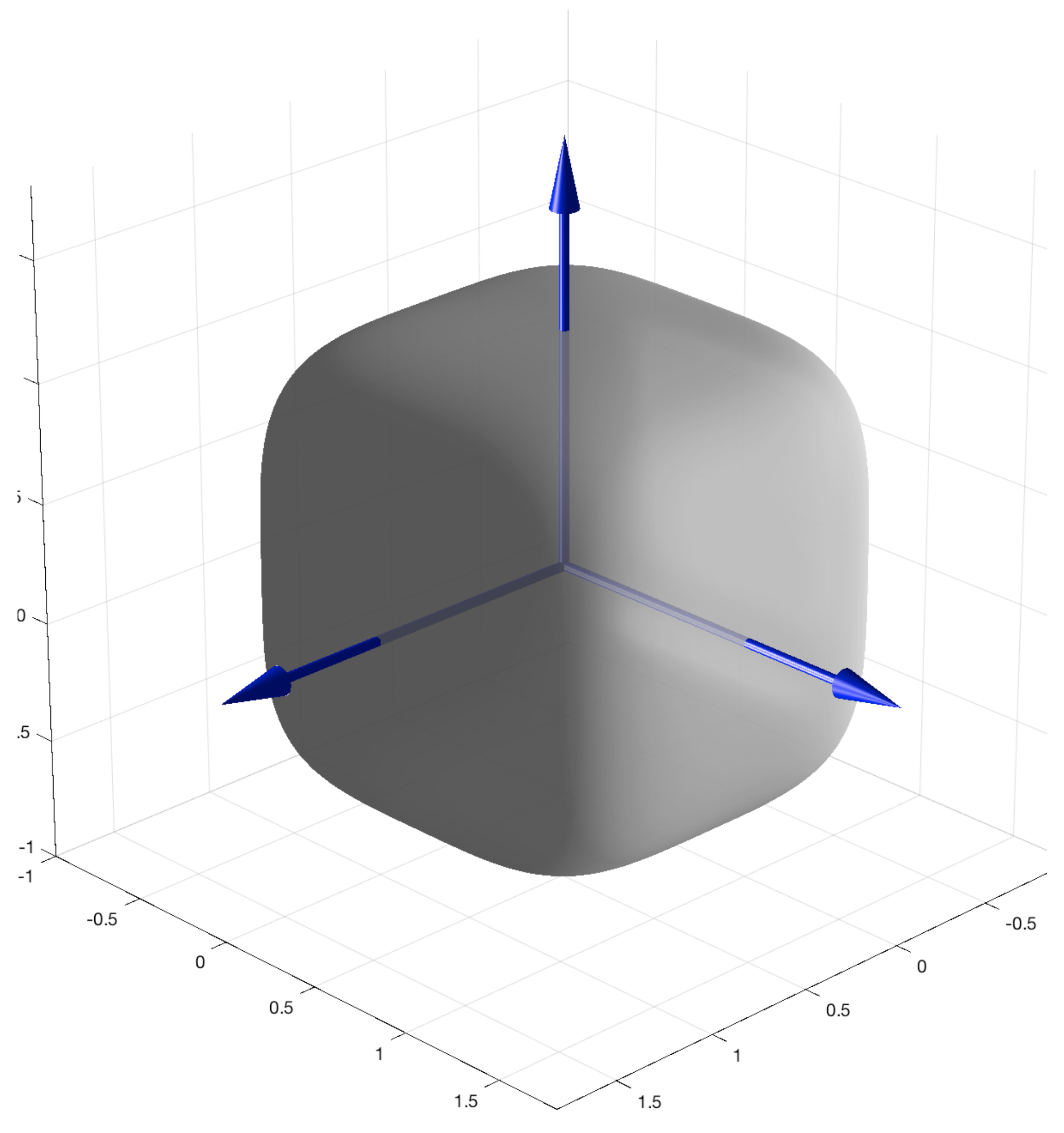}
    \includegraphics[width=.4\linewidth]{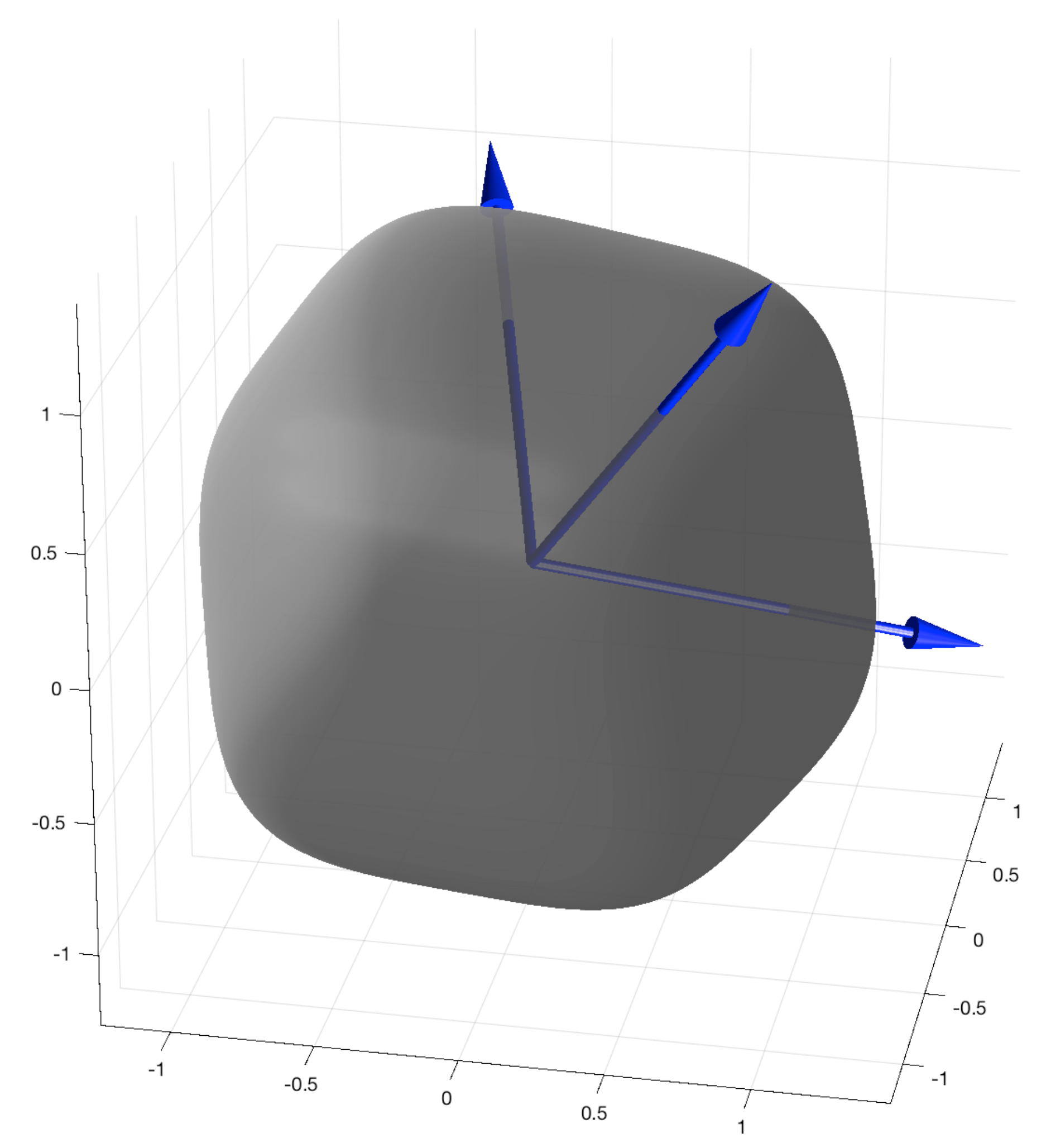}
  \end{center}
  \caption{Representation of the reference cross 
    as a dice-shaped polynomial isovalue surface (left),
    and of a general cross attitude as a rotation of the latter (right)
    \label{fig:1}}
\end{figure}

The idea promoted in this paper is thus also to work 
with polynomials 
whose isovalues exhibit the sought octahedral symmetry
but, instead of expanding them in a spherical harmonics basis,
they are represented as explicit rotations 
of a reference polynomial
\begin{equation}
  f_{ref}(x_1, x_2, x_3) = \|{x}\|_4^4 \equiv x_1^4 + x_2^4 + x_3^4,
  \label{reference_attitude}
\end{equation}
whose isovalue $f_{ref}=1$ is the dice-shaped surface  
depicted in Fig.~\ref{fig:1} (left).
Fourth order is the lowest polynomial order 
exhibiting distinctive octahedral symmetry,
which is by the way rather natural in a Cartesian coordinate system,
as it simply amounts to the invariance
against any argument inversion and/or permutation\,:
$$
f_{ref}(x_1,x_2,x_3) = f_{ref}(-x_1,x_2,x_3) = f_{ref}(x_2,-x_1,x_3) = \dots
$$

In tensor notations, we have
$$
f_{ref}(x_1, x_2, x_3) = \tilde{A}_{ijkl} \, x_i x_j x_k x_l
$$
assuming Einstein's implicit summation over repeated indices.
As a polynomial is characterized by its coefficients,
not by the power terms which act as a basis,
the $4^{th}$ order tensor
\begin{equation}
  \tilde{A}_{ijkl} = \sum_{q=1}^3 \delta_{iq} \delta_{jq} \delta_{kq} \delta_{lq}.
  \label{eq:1}
\end{equation}
is another full-fledged representation of the reference cross $f_{ref}$.
It has only three non-zero components
$$
\tilde{\At}_{1111} = \tilde{\At}_{2222} = \tilde{\At}_{3333} = 1.
$$

\subsection{Rotation of the reference cross}

The reference cross \eqref{reference_attitude} exhibits octahedral symmetry  
and rotations, which are isometries, preserve this symmetry.
It can therefore be stated that the space of all possible cross attitudes in $E^3$
is the set 
\begin{equation}
  f(x_1,x_2,x_3) = f_{ref}(R_{1i}x_i , R_{2j}x_j, R_{3k} x_k)
  \quad , \quad
  R_{ij} \in \mbox{SO}(3),
  \label{attitudePolynomial}
\end{equation}
whose corresponding tensor representation reads
\begin{equation}
  {\At}_{ijkl} = R_{im} R_{jn} R_{ko} R_{lp}\,\tilde{A}_{mnop} 
  = \sum_{q=1}^3  R_{im} R_{jn} R_{ko} R_{lp}\,\delta_{iq} \delta_{jq} \delta_{kq} \delta_{lq} 
  = \sum_{q=1}^3  R_{iq} R_{jq} R_{kq} R_{lq}.
  \label{eq:A}
\end{equation}
This tensor, noted $\At$, represents a general \emph{attitude} of the cross in $E^3$,
and the isovalue of the associated polynomial
$$
f(x_1,x_2,x_3) = {\At}_{ijkl} \,x_i x_j x_k x_l = 1
$$
is a rotation of the axis-aligned dice-shaped surface $f_{ref}$, Fig.~\ref{fig:1}.

Rotation matrices play a pivotal role in these definitions.
For convenience, let us define the following nomenclature\,:
\begin{itemize}
\item The indices $1, 2, 3$ refer to the angles $\alpha$, $\beta$ and $\gamma$, i.e. the angles corresponding to the first, second and third elemental rotations, respectively.
\item The matrices $X, Y, Z$ represent the elemental rotations
  about the axes $x_1, x_2, x_3$ of the Cartesian reference cross in $\Rset^3$  
  (e.g., $Y_1$ represents a rotation about $x_2$ by an angle $\alpha$).
\item The shorthands $s$ and $c$ represent sine and cosine 
  (e.g., $s_1$ represents the sine of $\alpha$).
\end{itemize}
We have for example the rotation matrix in $\Rset^3$ 
\begin{equation}
  R = Z_1 X_2 Z_3 = 
  \begin{pmatrix*}[r]
    c_1 c_3 - c_2 s_1 s_3 ~~      &  -c_1 s_3-c_3c_2s_1 ~~                   & s_2 s_1 \\
    c_1 c_2s_3 + c_3s_1 ~~&    c_1c_2c_3 -s_1 s_3~~  & -c_1 s_2 \\
    s_3 s_2~~ &    c_3s_2 ~~  & c_2
  \end{pmatrix*},
  \label{eq:R}
\end{equation}
which non-linearly depends on only three degrees of freedom\,:
the angles $\alpha$, $\beta$ and $\gamma$.

\subsection{Algebraic structure of ${\At}_{ijkl}$}

A $4$th order tensor in $E^3$ has at most $3^4 = 81$ independant components.
The specific algebraic structure of  \eqref{eq:A} makes it so, however,
that the tensor space of interest for cross fields is much smaller than that,
and can be characterized as a linear space $\tensspace$ of dimension $9$,
convenient for finite element interpolation,
together with a non-linear projection operator 
\begin{equation}
  \Pi\,: \tensspace  \mapsto \mbox{SO}(3) / \mbox{O}
\end{equation}
from the 9-dimensional linear space 
onto a 3-dimensional nonlinear manifold. 

The demonstration of this algebraic structure is in several steps.
First, the number of independent components of $\At$ 
cannot be larger than the dimension 
of the space of homogeneous polynomials of order $4$ 
with $3$ variables, i.e., ${{4+3-1} \choose {4}} = 15$. 
This is a mere consequence of the fact that the products of coordinates,
as the product of real numbers, obviously commute,
$x_i x_j = x_j x_i$,
and that all terms associated with components of $\At$
whose indice sets are permutations of each other
eventually contribute to the same term in the polynomial. 
In mathematical terms, it amounts to require 
the tensor $\At$ be \emph{fully symmetric},
a condition usually written 
$$
\At_{ijkl} = \At_{(ijkl)}  
$$
with 
$$
\At_{(ijkl)}  = \frac{1}{24} \left( 
\At_{ijkl}  + \At_{jikl}  + \dots
\right)
$$
where the 24 permutations of the set $ijkl$
are enumerated at the right-hand side. 

The tensors \eqref{eq:A} have however deeper structures,
related with the unitary property 
$$
R^t R = I
\quad , \quad
R_{ik} R_{jk} = \delta_{ij}
$$
of  rotations matrices.
The so called ``partial traces''\footnote{Make sure to clearly distinguish
  sums over two repeated indices, which are implict in our notation,
  and sums over four repeated indices, which are explicitly written.} 
\begin{equation}
  {\At}_{iikl} = \sum_{q=1}^3 R_{iq}R_{iq}R_{kq}R_{lq} 
  = \sum_{q=1}^3\delta_{qq}R_{kq}R_{lq} 
  = R_{kq}R_{lq} = \delta_{kl}
  \label{eq:traces}
\end{equation}
leads to the $6$ additional relationships
\begin{eqnarray}
  {\At}_{1111} + {\At}_{2211} +  {\At}_{3311} &=& 1, \nonumber \\
  {\At}_{1122} + {\At}_{2222} +  {\At}_{3322} &=& 1, \nonumber \\
  {\At}_{1133} + {\At}_{2233} +  {\At}_{3333} &=& 1, \nonumber\\
  {\At}_{1112} + {\At}_{2212} +  {\At}_{3312} &=& 0, \nonumber \\
  {\At}_{1113} + {\At}_{2213} +  {\At}_{3313} &=& 0, \nonumber \\
  {\At}_{1123} + {\At}_{2223} +  {\At}_{3323} &=& 0.  \nonumber
\end{eqnarray}
Other partial traces would give linearly dependent relationships,
due to the full symmetry of the tensor mentioned above. 

It is important to note that partial traces
are conserved under 
affine combination of tensors. 
Tensors in $\At$ form thus a $15-6=9$ dimensional linear space,
noted $\tensspace$,  convenient for finite element interpolation. 
Interestingly enough, this dimension is also that 
of the the space of $4$th order spherical harmonics, 
used by some authors to represent crosses \cite{ray2016practical}. 

\subsection{$\At$ is a projector}

Let $X$ be the set of symmetric $2^d$ order tensors in $E^3$.
This is a linear space,
and any tensor $\vd \in X$ can be expanded in terms of 
rank one basis tensors
$$
\vd = d_{mn} \ e_m \otimes e_n  
\quad , \quad 
d_{mn} = d_{nm}
$$
where $e_m$,  $m,1,2,3$, are the orthormal basis vectors of the Cartesian coordinate system. 
Alternatively, basis tensors rotated by a matrix $R \in SO(3)$ can be used as well,
\begin{equation}
  \vd = d'_{mn} \ r_m \otimes r_n
  \quad , \quad 
  d'_{mn} = d'_{nm}
  \label{eq:rankoneexp}
\end{equation}
with now
$$
r_m = R(e_m) 
\quad , \quad 
(r_m)_i = R_{ij} \delta_{jm} = R_{im}
$$
the $m^{th}$ column vector of the rotation matrix $R$. 
The polynomials we are using in this paper to represent cross attitudes
are built from special tensors in $X$
for which one simply has $B'_{mn} = x_m x_n$.

Once the space $X$ is appropriately characterized, 
the tensor $\At$ defined by \eqref{eq:A}
can be regarded as a linear application
$$
\At\,: X \mapsto X,
$$
and it is readily shown that it is a projection operator\,:
\begin{eqnarray}
  \label{eq:proj}
  \At^2 = \At:\At&\equiv& {\At}_{ijmn}{\At}_{mnkl} \nonumber \\
  &=& \sum_{q=1}^3 \sum_{s=1}^3 R_{iq}R_{jq}R_{mq}R_{nq}R_{ns}R_{ms}R_{ks}R_{ls} \nonumber \\
  &=& \sum_{q=1}^3 \sum_{s=1}^3 R_{iq}R_{jq}
  R_{ks}R_{ls}
  \underbrace{R_{mq} R_{ms}}_{\delta_{qs}}
  \underbrace{R_{nq} R_{ns}}_{\delta_{qs}} \nonumber \\
  &=&  \sum_{q=1}^3 R_{iq}R_{jq} R_{kq}R_{lq} =  {\At}_{ijkl} = \At,
\end{eqnarray}

To characterize this projection,
the image by  $\At$ of the basis tensors $\ r_m \otimes r_n$ in \eqref{eq:rankoneexp}
is evaluated. One has
\begin{eqnarray}
  \left(\At \,: \, r_m \otimes r_n\right) \big|_{ij} 
  &=& {\At}_{ijkl} \ R_{km} R_{ln} \nonumber \\
  &=& \sum_{q=1}^3 R_{iq} R_{jq} R_{kq} R_{lq} R_{km} R_{ln} \nonumber \\
  &=& \sum_{q=1}^3 R_{iq} R_{jq} \delta_{qm} \delta_{qn},  \nonumber
\end{eqnarray}
from where follows
\begin{eqnarray}
  \At \,: ( r_m \otimes r_m) &=  r_m \otimes r_m &m=1,2,3 \mbox{ (no sum)}\\
  \At \,: ( r_m \otimes r_n) &=  0 &\mbox{ if }  m \ne n.
\end{eqnarray}
As expected for a projector, 
eigen values are either 0 or 1. 
The eigenspace corresponding to
the eigenvalues $\lambda_1 = \lambda_2 = \lambda_3 = 1$ is 
$$
\mbox{range\,}\At =  \mbox{span\,}(r_1 \otimes r_1, r_2 \otimes r_2, r_3 \otimes r_3) \subset X
$$
whereas that corresponding to $\lambda_4 = \lambda_5 = \lambda_6 = 0$ is
the kernel space
$$
\ker\At =
\mbox{span\,} \left(
r_1 \otimes r_2 + r_2 \otimes r_1, 
r_2 \otimes r_3 + r_3 \otimes r_2, 
r_3 \otimes r_1 + r_1 \otimes r_3 \right).
$$
Obviously, the \emph{eigentensors} $\vd^j$ are both
symmetric ($\vd^j_{mn} = \vd^j_{nm}$),
and orthonormal to each other ($\vd^i : \vd^j = \delta_{ij}$)
under the Frobenius norm $\|\ \vd\ \|^2_F = \vd : \vd  = d_{mn} d_{nm}$.

\section{Three main results}

The three main results of the paper are now presented in this section. 
We first prove that any tensor $\At \in \tensspace$
(a fully symmetric tensor obeying \eqref{eq:traces})
corresponds to a cross attitude (i.e., a rotation of the reference cross)
if it is a projector $\At^2 = \At$ onto a three-dimensional subspace of $X$.
Then, we show how to project a tensor $\At \in \tensspace$ 
that is not a projector 
onto another tensor in $\tensspace$ verifying $\At^2 = \At$
with three non-zero eigenvalues.
Finally, we show the direct relashionship between spherical harmonics
and our representation on terms of $4^{th}$ order tensors.

\subsection{Sufficiency}

\begin{theorem}
  \label{theorem::sufficiency2}
  A tensor $\At \in \tensspace$,
  (fully symmetric $4^{th}$ order tensor 
  obeying the partial trace condition\eqref{eq:traces})
  that is also a projector on a 3-dimensional subspace of $X$
  corresponds to a cross attitude (i.e., to a rotation of the reference cross)
\end{theorem}

\begin{proof}
  If $\At$ is a projector onto a 3-dimensional subspace of $X$,
  there exist three orthonormal symmetric second order tensors 
  $\vd^a,\vd^b,\vd^c \in X$ such that
  \begin{equation}
    \begin{array}{rcll}
      \vd^a\otimes \vd^a + \vd^b\otimes \vd^b + \vd^c\otimes \vd^c &=& \At & \\
      \vd^l : \vd^m &=& \delta_{lm}&l,m=a,b,c\\
      \vd^l_{ij} &=& \vd^l_{ji}&l=a,b,c,~ i,j=1,2,3.\\
      \label{eq:propD}
    \end{array}
  \end{equation}
  Note that there is no implicit summation on upper indices in this proof. 
  The key point of the proof is to show that 
  the eigentensors $\vd^a$,  $\vd^b$ and $\vd^c$ commute with each other.
  If this is the case, they are joint diagonalizable 
  and share therefore the same set of eigenvectors.
  It is then easy to see that $\At$ is the $4^{th}$ order tensor representation of a cross. 

  Let
  $$
  {\ctr lm} = \vd^l\cdot\vd^m-\vd^m\cdot\vd^l
  $$
  be the the commutator of $\vd^l$ and $\vd^m$,
  of which we have to prove  the Frobenius norm is zero,
  $$
  \|\ {\ctr lm}\ \|^2_F = {\ctr lm}:{\ctr lm} =0.
  $$

  One first notes that
  \begin{eqnarray}
    (\vd^a \cdot \vd^b) : (\vd^e \cdot \vd^f) 
    &=& d^a_{ik} d^b_{kj} d^e_{il} d^f_{lj} = \mbox{tr\,}(\vd^a\vd^b\vd^f\vd^e) \nonumber\\
    &=& (\vd^b \cdot \vd^a) : (\vd^f \cdot \vd^e) \nonumber\\ 
    &=& (\vd^a \cdot \vd^e) : (\vd^b \cdot \vd^f) \nonumber\\
    &=& (\vd^e \cdot \vd^a) : (\vd^f \cdot \vd^b) \label{eq:reorganizations}
  \end{eqnarray}
  exploiting the symmetry of the individual $\vd^l$ tensors, 
  and all possible reorganizations of the matrix products.
  As the contraction operator $:$ is also symmmetric,
  there are thus 8 equivalent argument permutations (out of 24) 
  for such scalar quantities. 
  With this, one shows that
  \begin{eqnarray}
    \|\ [\vd^a, \vd^b] \ \|^2_F 
    &=& ( \vd^a \cdot \vd^b - \vd^b \cdot \vd^a) 
    \,: ( \vd^a \cdot \vd^b - \vd^b \cdot \vd^a) \nonumber\\
    &=& (\vd^a \cdot \vd^b) \,: (\vd^a \cdot \vd^b)  
    - (\vd^a \cdot \vd^b) \,: (\vd^b \cdot \vd^a) -\nonumber\\
    &&      (\vd^b \cdot \vd^a) \,: (\vd^a \cdot \vd^b)
    + (\vd^b \cdot \vd^a) \,: (\vd^b \cdot \vd^a)  \nonumber\\
    &=& 2 (\vd^a \cdot \vd^a) \,: (\vd^b \cdot \vd^b) 
    -  2 (\vd^a \cdot \vd^b) \,: (\vd^b \cdot \vd^a) \label{eq:commutator}
  \end{eqnarray}
  the last two terms being not reducible to each other 
  by the permutation rules given above. 
  As
  $$
  (\vd^a \cdot \vd^b) : (\vd^e \cdot \vd^f) = \mbox{tr\,}(\vd^a\vd^b\vd^f\vd^e),
  $$
  the identity \eqref{eq:commutator} can also be interpreted as 
  \begin{equation}
    \|\ [\vd^a, \vd^b] \ \|^2_F 
    = \mbox{tr\,}(\vd^a\vd^a\vd^b\vd^b - \vd^a\vd^b\vd^a\vd^b)
    = \mbox{tr\,}\left( (\vd^a)^2 (\vd^b)^2 - (\vd^a\vd^b)^2 \right).
    \label{eq:frobeniustrace}
  \end{equation}

  The identity tensor $I$ being in the range of $\At$,
  it is an eigen tensor of $\At$, one has thus
  $$
  \delta_{ij} = A_{ijkl} \delta_{kl} = A_{iklj} \delta_{kl} 
  $$
  where the full symmetry of $\At$ has been used.
  This reads, without components,
  $$
  I = \vd^a \vd^a + \vd^b \vd^b + \vd^c \vd^c
  = (\vd^a)^2 + (\vd^b)^2 + (\vd^c)^2,
  $$
  wherefrom directly follows 
  \begin{equation}
    (\vd^a)^2 = (\vd^a)^4 + (\vd^a)^2 (\vd^b)^2 + (\vd^a)^2 (\vd^c)^2.
    \label{eq:dasquare1}
  \end{equation}

  On the other hand, using now the fact that $\vd^a$ is an eigen tensor of $\At$, 
  one has
  $$
  d^a_{ij}  = A_{ijkl} d^a_{kl} 
  = d^a_{ij} d^a_{kl} d^a_{kl} + d^b_{ij} d^b_{kl} d^a_{kl} + d^c_{ij} d^c_{kl} d^a_{kl}
  $$
  and, using again the full symmetry of $\At$
  $$
  d^a_{ij} = d^a_{ik} d^a_{lj} d^a_{kl} + d^b_{ik} d^b_{lj} d^a_{kl} + d^c_{ik} d^c_{lj} d^a_{kl}
  $$
  so that
  $$
  \vd^a = (\vd^a)^3 + \vd_b \vd_a \vd_b + \vd_c \vd_a \vd_c 
  $$
  and premultiplying with $\vd^a$
  \begin{equation}
    (\vd^a)^2 = (\vd^a)^4 + (\vd_a \vd_b)^2 + (\vd_a \vd_c)^2.
    \label{eq:dasquare2}
  \end{equation}

  Substraction of \eqref{eq:dasquare2} and \eqref{eq:dasquare1} yields
  $$
  0 = (\vd_a \vd_b)^2 + (\vd_a \vd_c)^2 - (\vd^a)^2 (\vd^b)^2 + (\vd^a)^2 (\vd^c)^2,
  $$
  the trace of which gives, using \eqref{eq:frobeniustrace},
  $$
  0 = \|\ [\vd^a, \vd^b] \ \|^2_F  + \|\ [\vd^a, \vd^c] \ \|^2_F.
  $$
  As this is a sum of positive terms, both terms are zero, 
  and we have proven that $\vd^a$ commutes with $\vd^b$ and $\vd^c$.

  As $(\vd^a,\vd^b,\vd^c)$ are symmetric and commute, there exist an othonormal basis $(r^1,r^2,r^3)\in(\mathbb{R}^3)^3$ such as:
  \begin{equation}
    \begin{array}{rcll}
      \vd^a &=& \alpha_1r^1\otimes r^1+\alpha_2r^2\otimes r^2+\alpha_3r^3\otimes r^3 & \alpha_1,\alpha_2,\alpha_3\in\mathbb{R}\\
      \vd^b &=& \beta_1r^1\otimes r^1+\beta_2r^2\otimes r^2+\beta_3r^3\otimes r^3 & \beta_1,\beta_2,\beta_3\in\mathbb{R}\\
      \vd^c &=& \gamma_1r^1\otimes r^1+\gamma_2r^2\otimes r^2+\gamma_3r^3\otimes r^3 & \gamma_1,\gamma_2,\gamma_3\in\mathbb{R}\\
    \end{array}
    \label{eq:diag1TP}
  \end{equation}

  We will now show that $r_i\otimes r_i\text{, }i\in \{1,2,3\}$ are eigentensors of $\At$.
  First, we know that $\vd^l$ are orthogonal and of norm $1$. So, $\left((\alpha_1,\alpha_2,\alpha_3),(\beta_1,\beta_2,\beta_3),(\gamma_1,\gamma_2,\gamma_3)\right)$ forms an orthonormal basis of $\mathbb{R}^3$.

  Therefore, it exists a unique vector $v\in\mathbb{R}^3$ such as :
  \begin{equation}
    \begin{pmatrix*}[r]
      \alpha_1 ~~     &  \beta_1 ~~                   & \gamma_1 \\
      \alpha_2 ~~     &  \beta_2 ~~                   & \gamma_2 \\
      \alpha_3 ~~     &  \beta_3 ~~                   & \gamma_3
    \end{pmatrix*}
    v
    =
    \begin{pmatrix*}[r]
      1  \\
      0  \\
      0 
    \end{pmatrix*}
    \label{eq:CL1TP}
  \end{equation}
  and we have
  \begin{equation}
    v_1\vd^a + v_2\vd^b + v_3\vd^c = r_1\otimes r_1
    \label{eq:CL2TP}
  \end{equation}
  As, $\vd^l$ are eigentensors of $\At$ assiociated to eigenvalue $1$,
  \begin{equation}
    \begin{array}{rcl}
      \At : (r_1\otimes r_1) &=& \At : (v_1\vd^a + v_2\vd^b + v_3\vd^c)\\
      &=& v_1\vd^a + v_2\vd^b + v_3\vd^c\\
      &=& (r_1\otimes r_1)\\
    \end{array}
    \label{eq:eige1TP}
  \end{equation}
  Consequently, $r_1\otimes r_1$ is an eigentensor of $\At$ assiociated to eigenvalue $1$. We can show in the same way that $(r_2\otimes r_2)$ and $(r_3\otimes r_3)$ are also eigentensors of $\At$ assiociated to eigenvalue $1$.

  Thus, as $\At$ is a projector with only three non zero eigenvalues, we finally have :
  \begin{equation}
    \begin{array}{rcl}
      \At &=&  r_1\otimes r_1\otimes r_1\otimes r_1 + r_2\otimes r_2\otimes r_2\otimes r_2 + r_3\otimes r_3\otimes r_3\otimes r_3\\
    \end{array}
    \label{eq:diag2TP}
  \end{equation}
  Therefore, $\At$ is the representation of the cross with orthogonal directions $(r_1,r_2,r_3)$.
\end{proof}

\subsection{Recovery}
The representation that is advocated here relies heavily on 
the computation of eigentensors of fourth order
tensors. Disappointingly, numerical tools for linear algebra are
designed to manipulate vectors and matrices. Hopefully, it is possible
to represent symmetric fourth order tensors as matrices.

A  fourth order tensor $\At$  endowed with minor symmetry conditions
${\displaystyle {\At}_{ijkl}={\At}_{jikl}={\At}_{ijlk}}$ has $36$ independant
components. It is useful to write it in the so called Mandel notation
as the following matrix $6 \times 6$ matrix:
\begin{equation}
  {A}={\begin{pmatrix}{\At}_{{1111}}&{\At}_{{1122}}&{\At}_{{1133}}&{\sqrt
        2}{\At}_{{1123}}&{\sqrt 2}{\At}_{{1113}}&{\sqrt
        2}{\At}_{{1112}}\\
      {\At}_{{2211}}&{\At}_{{2222}}&{\At}_{{2233}}&{\sqrt
        2}{\At}_{{2223}}&{\sqrt 2}{\At}_{{2213}}&{\sqrt 2}{\At}_{{2212}}\\
      {\At}_{{3311}}&{\At}_{{3322}}&{\At}_{{3333}}&{\sqrt
        2}{\At}_{{3323}}&{\sqrt 2}{\At}_{{3313}}&{\sqrt 2}{\At}_{{3312}}\\
      {\sqrt 2}{\At}_{{2311}}&{\sqrt 2}{\At}_{{2322}}&{\sqrt
        2}{\At}_{{2333}}&2{\At}_{{2323}}&2{\At}_{{2313}}&2{\At}_{{2312}}\\
      {\sqrt 2}{\At}_{{1311}}&{\sqrt 2}{\At}_{{1322}}&{\sqrt
        2}{\At}_{{1333}}&2{\At}_{{1323}}&2{\At}_{{1313}}&2{\At}_{{1312}}\\
      {\sqrt 2}{\At}_{{1211}}&{\sqrt 2}{\At}_{{1222}}&{\sqrt
        2}{\At}_{{1233}}&2{\At}_{{1223}}&2{\At}_{{1213}}&2{\At}_{{1212}}\\
  \end{pmatrix}}.
  \label{AVoigt}
\end{equation}
Major symmetry conditions ${\displaystyle {\At}_{ijkl}={\At}_{klij}}$
ensure that $A$ is symmetric.
Factors $2$ and $\sqrt{2}$ in \eqref{AVoigt} allow to write 
the cross representation as the following usual quadratic form:
\begin{equation}
  (x \otimes x)^t \ {A}  \ (x \otimes x) = 1.
  \label{xtxAxtx}
\end{equation}
with
$$
x \otimes x = 
\begin{pmatrix*}[l]
  x_1^2 ~~ &x_2^2 ~~ &x_3^2  ~~&\sqrt{2}x_2x_3 ~~& \sqrt{2}x_1x_3  ~~&\sqrt{2}x_1x_2 \\
\end{pmatrix*}^t.
$$
Let us now compute Mandel's representation of the reference cross
$\tilde{\At}$ (see \eqref{reference_attitude}):
\begin{equation}
  \tilde{A}  
  = \begin{pmatrix*}[r]
    1 &0 &0 & 0 & 0 & 0\\
    0 &1 &0 & 0 & 0 & 0\\
    0 &0 &1 & 0 & 0 & 0\\
    0 &0 &0 & 0 & 0 & 0\\
    0 &0 &0 & 0 & 0 & 0\\
    0 &0 &0 & 0 & 0 & 0
  \end{pmatrix*}
  \label{Atilde}
\end{equation}
In a previous section, we have shown that 
only $9$ scalar parameters $(a_1,\dots,a_9)$
are required to represent $\At$. Taking into account symmetries 
and partial traces, we can write
\begin{equation}
  {\tiny
    {A} = 
    \begin{pmatrix*}[r]
      a_{1} & & &  &  & \\
      \frac{1}{2} (1 + a_{3} - a_{2} - a_{1})& a_{2} & & & \mbox{SYM}  &  \\
      \frac{1}{2} (1 + a_{2} - a_{3}- a_{1})& \frac{1}{2} (1 +
      a_{1} - a_{2}- a_{3})& a_{3} & & &  \\
      -\sqrt{2} (a_{4}+a_5)& \sqrt{2} a_{4} & \sqrt{2} a_{5}&    1 + a_{1} - a_{3} - a_{2}& &  \\
      \sqrt{2} a_{6}& -\sqrt{2} (a_6+a_{7})& \sqrt{2} a_{7}& -{2}(a_{8}+a_{9})& 1 + a_{2} - a_{3} - a_{1} &  \\
      \sqrt{2} a_{8}& \sqrt{2} a_{9}& -\sqrt{2} (a_{8}+a_{9})&{-2} (a_6+a_{7}) & {-2}(a_{4}+a_5) & 1 + a_{3} - a_{2} - a_{1} \\
    \end{pmatrix*}
    \label{eq:mandel}
  }
\end{equation}
with the following correspondances between the ${\At}_{ijkl}$'s and the $a_i$'s:
$$a_1= {\At}_{1111}~, a_2 = {\At}_{2222},~a_3= {\At}_{3333},$$
$$a_4 = {\At}_{2322}~, a_5 = {\At}_{2333},~a_6 = {\At}_{1311},$$
$$a_7 = {\At}_{1333}~, a_8 = {\At}_{1211},~a_9 = {\At}_{1222}.$$

Mandel's notation allows to write tensor contractions as matrix
products. For example, $\At : \At = \At$ (see \eqref{eq:proj}) is written using Mandel's notation as
$A \cdot A = A$. 

Eigenvectors of $A$ are the eigentensors of $\At$. Their $6$
components are the $6$ independent entries of eigentensors 
$\vt^k$ that are symmetric second order tensors. The two following
MATLAB functions allow to transform fourth order tensors $\At$ into
Mandel's form and transform eigenvectors of $A$ into second order
tensors. We also see factors of $\sqrt{2}$ that accounts for the
symmetry of $A$.
{\scriptsize
\begin{verbatim}
  function D = Vec6ToTens2 (v)
  s = 2^(1./2.);
  D = [
    v(1) , v(6)/s , v(5)/s ; 
    v(6)/s , v(2) , v(4)/s ; 
    v(5)/s , v(4)/s , v(3) ; 
  ];
  end
\end{verbatim}
}
{\scriptsize
\begin{verbatim}
  function a = Tens4ToMat6 (A)
  s = 2^(1./2.);
  a = [ 
    A(1,1,1,1) , A(1,1,2,2) , A(1,1,3,3) , s*A(1,1,2,3), s*A(1,1,1,3), s*A(1,1,1,2) ;
    A(2,2,1,1) , A(2,2,2,2) , A(2,2,3,3) , s*A(2,2,2,3), s*A(2,2,1,3), s*A(2,2,1,2) ;
    A(3,3,1,1) , A(3,3,2,2) , A(3,3,3,3) , s*A(3,3,2,3), s*A(3,3,1,3), s*A(3,3,1,2) ;
    s*A(2,3,1,1) , s*A(2,3,2,2) , s*A(2,3,3,3) , 2*A(2,3,2,3), 2*A(2,3,1,3), 2*A(2,3,1,2) ;
    s*A(1,3,1,1) , s*A(1,3,2,2) , s*A(1,3,3,3) , 2*A(1,3,2,3), 2*A(1,3,1,3), 2*A(1,3,1,2) ;
    s*A(1,2,1,1) , s*A(1,2,2,2) , s*A(1,2,3,3) , 2*A(1,2,2,3), 2*A(1,2,1,3), 2*A(1,2,1,2) ;
  ];
  end
\end{verbatim}
}
Note that those MATLAB routines are made for testing and that 3D large
codes will only manipulate the $9$ nodal unknowns $a_i$.

Computing 3D cross fields implies to propagate tensors that have known
values on the boundary of a 3D domain inside the domain. Assume a
tensor $\At$ that has the right structure and that is such that $\At
: \At = \At$. With such properties, we know that $\At$ is a rotation
of $\tilde{\At}$. The first important issue is about backtracking $R$
from $\At$ i.e. find the three orthonormal column vectors $r^q$ of $R$
that  form $\At$ through Equation \eqref{eq:A}.

An eigentensor $\vt^n$ of $\At$ that is associated with eigenvalue $1$ is 
the linear combination 
$$\vt^n_{ij} = \sum_{q=1}^3 c_{nq} r^q_i r^q_j.$$
We have
$$\vt^n_{im}r^k_m = c_{nk} r^k_i$$
which means that the eigenvectors of $\vt^n$ are indeed the $r^k$'s. One
issue here could be that $\vt^n$ is not of full rank. Yet, the sum
$$\vt = \sum_{k=1}^3 \vt^k$$
is of full rank. Eigenvectors of $\vt$ are the wanted $3$
directions. 
Assume a representation $\At$ in of
the form \eqref{eq:mandel} and let us recover rotation matrix $R$. The
following MATLAB code recovers  the rotation matrix $R$ starting from
a tensor $\At$ that is a rotation.
{\scriptsize
\begin{verbatim}
  function R = Tens4ToRotation (A)
  a     = Tens4ToMat6 (A); % transform A into its matrix form 
  [V,D] = eig (a)   ; % compute eigenspace
  [X,I] = sort(diag(D)); % sort eigenvalues
  % compute the sum of eigentensors of A associated 
  % to eigenvalues equal to 1
  V2      = Vec6ToTens2 (V(:,I(4))+V(:,I(5))+V(:,I(6))); 
  [R,d2] = eig(V2); % get rotation matrix R
  end
\end{verbatim}
}
This code has been tested to thousands of random rotations, giving the
right answer in a $100 \%$ robust fashion.

The aim of our work is to build smooth cross fields in general 3D
domains. For that, we will solve a boundary value problem
for the $9$ linearly independant components $(a_1,\dots,a_9)$ of 
the tensor representation. 
Consider two representations $X$ and $Y$ with their representation
vectors $(x_1,\dots,x_9)$ and $(y_1,\dots,y_9)$
Any smoothing procedure computes (weighted) averages of such
representations. For example, representation vector
$${1 \over 2} (x_1+y_1,\dots,x_9+y_9)$$
allows to build Mandel's representation 
$Z = {1 \over 2}(X + Y)$ that as the same structure as 
matrix \eqref{eq:mandel}. 

Assume a cross attitude $A(\alpha,\beta,\gamma)$ that depends
on Euler angles $\alpha$, $\beta$ and $\gamma$. The projection of $Z$
into the space of rotations of the reference cross is 
defined as the cross attitude $A$ that verifies
$$A = \min_{\alpha,\beta,\gamma} \|A(\alpha,\beta,\gamma)-Z\|.$$ 
The following function
{\scriptsize
\begin{verbatim}
  function P = projection (A) 
  b_guess = [0 0 0];
  [b_guess(1) b_guess(2) b_guess(3)] = EulerAngles(Tens42Rotation (A));
  vA = Tensor4ToMat6 (A); 
  fun = @(x) norm(Tensor4ToMat6(makeTensor (makeEulerRotation (x(1),x(2),x(3))))-vA) ;
  b_min = fminsearch(fun, b_guess);
  P = makeTensor (makeEulerRotation (b_min(1), b_min(2),b_min(3)));
  end
\end{verbatim}
}
allows to compute such a projection. In that function, we choose as an
initial guess for Euler angles the value computed by {\tt
  Tens4ToRotation} which uses the eigenspace of $A$ relative to its
three largest eigenvalues. Figure \ref{fig:comp} shows that this
initial guess is indeed a very good approximation of the
projection. In reality, it is such a good approximation that it can be
used as is without doing the exact minimization.  
\begin{figure}[H]
  \begin{center}
    \includegraphics[width=.5\textwidth]{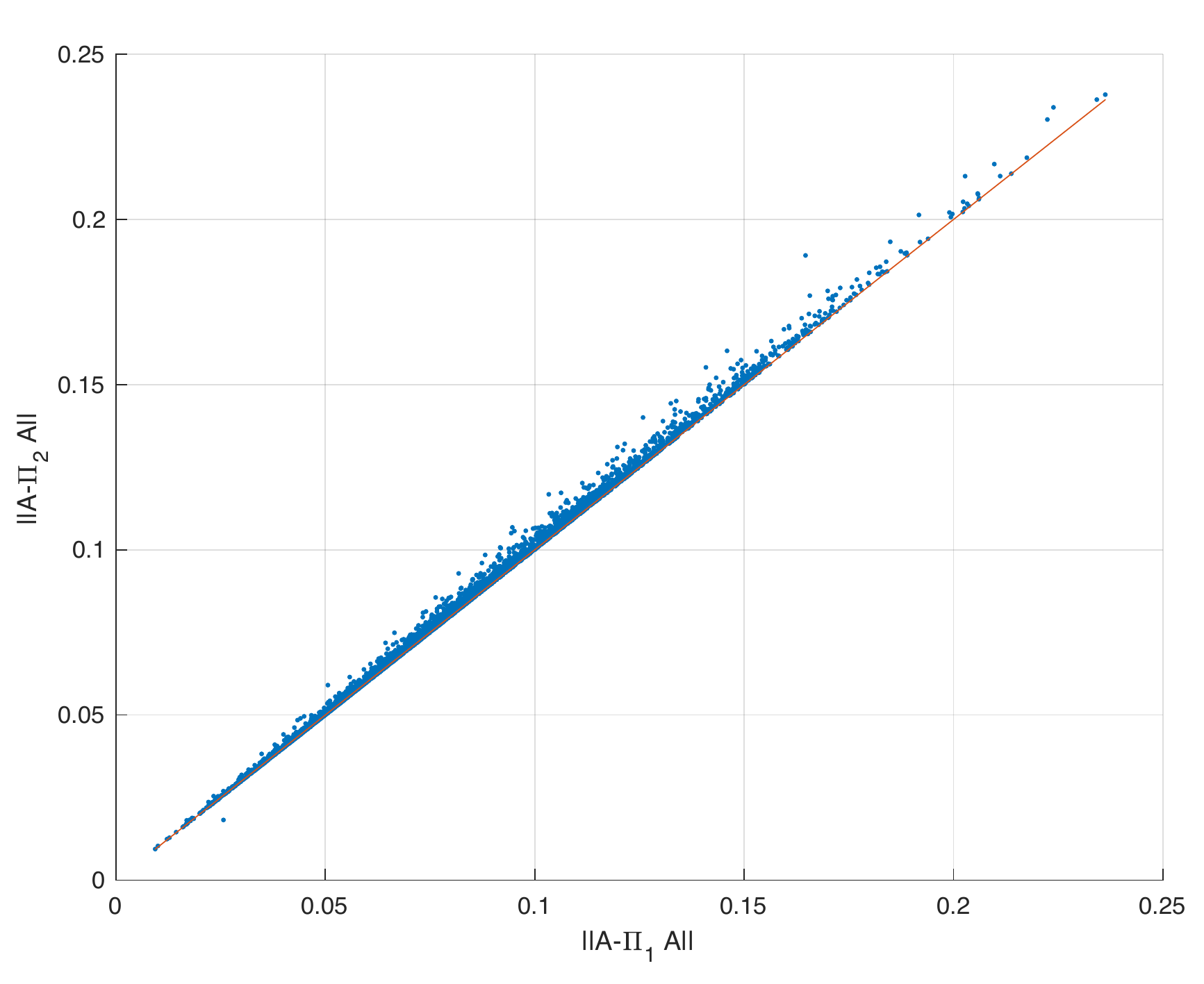}
  \end{center}
  \caption{Projection of $4000$ random tensors $Z$. $P_1$ is the true
    projection while $P_2$ is the approximation computed using function 
    {\tt Tens42Rotation}.
    \label{fig:comp}}
\end{figure}
On Figure \ref{fig:comp}, $P_1$ is the exact projection while $P_2$ is
the approximation. We see that the approximation $P_2$ is always very
good with respect to the projection, while being extremely simple and
fast to compute.

\subsection{Relation with spherical harmonics}
Harmonic polynomials $h(x)$ are polynomials that are such $\nabla^2 h
= 0$. Consider the rotated diced cube polynomial representation
$$\alpha(x) = \sum_{q=1}^3 (r^q \cdot x)^4$$
We have
$$\nabla^2 \alpha = \sum_{j=1}^3{\partial^2 \alpha \over \partial x^2_j} = 12 \sum_{q=1}^2 (r^q \cdot
x)^2 (r^q_j)^2 = 12 \sum_{q=1}^3 \left[(r^q \cdot x)^2
  \left((r^q_1)^2+(r^q_2)^2+(r^q_2)^2\right)\right] = 12 \sum_{q=1}^3 (r^q \cdot x)^2. $$
The equation $\sum_{q=1}^3 (r^q \cdot x)^2$ is the one of the unit
sphere that is invariant by rotation. Thus,
$$\nabla^2 \alpha = 12 \left| x\right|^2.$$ 
Representation polynomial $\alpha(x)$ is thus not harmonic. Yet, acknowledging that
$$\nabla^2 |x|^4  = 20 |x|^2,$$
we can define the following projection operator of diced cubes onto harmonic polynomials
$$P_{\shspace}(\alpha(x)) = \alpha(x) - {3 \over 5} |x|^4.$$
Operator $P_{\shspace}$ essentially remove three fifth of a sphere
to the diced cube so that the representation retains its symmetry
properties while becoming itself harmonic. Let us show that 
$P_{\shspace}$ is an orthogonal projector with respect to a norm that is related
to spherical harmonics. Consider the unit sphere $S^2$ and compute
$$\int_{S^2} h(x) \left[P_{\shspace}(\alpha(x)) - \alpha(x) \right] dx =
-{3 \over 5} \int_{S^2} h(x) |x|^4 dx=  -{3 \over 5}\int_{S^2} h(x) dx.$$ 
Harmonic functions are endowed with the mean value property which
states that the average of $h(x)$ over any sphere centered at $c$ is
equal to $h(c)$. So, 
$$\int_{S^2} h(x) \left[P_{\shspace}(\alpha(x)) - \alpha(x) \right] dx =
-{3 \over 5}\int_{S^2} h(0) dx.$$ 
Harmonic polynomials are homogeneous so $h(0)=0$ and operator
$P_{\shspace}$ is an orthogonal projector onto fourth order 
spherical harmonics.

As an example, consider our reference diced cube  that is represented
by $\alpha(x) = x_1^4+x_2^4+x_3^4$. Its  projection onto $\shspace$ is
$$P_{\shspace}(x_1^4+x_2^4+x_3^4) = {2\over 5} (x_1^4+x_2^4+x_3^4 - 3 (x_1^2 x_2^2+x_1^2 x_3^2+x_2^2 x_3^2))$$
This is indeed interesting to see that, for $x \in S^2$, we have
$$P_{\shspace}(x_1^4+x_2^4+x_3^4) = \sqrt{12 \pi \over 7} {16 \over 3} \left(\sqrt{7\over 12} Y_{4,0} + \sqrt{5\over 12}  Y_{4,4}\right) $$
where $Y_{4,j}$, $j=-4,\dots,4$ are the orthonormalized real spherical
harmonics. In \cite{ray2016practical}, authors define their reference frame as
$$\tilde{F} = \sqrt{7\over 12} Y_{4,0} + \sqrt{5\over 12} Y_{4,4}$$ 
which is to a constant the orthogonal projection of our reference
frame onto spherical harmonics.

\section{Practical computations}
Assume a domain $\Omega$ with its non smooth boundary $\Gamma$ that may 
contain sharp edges and corners. Our aim is to find a crossfield $F$ that is smooth, and such as for all $x\i\Gamma$, $F(x)$ has one direction aligned to the boundary normal $n(x)$. Here, a
simple smoothing procedure that consist in locally averaging cross
attitudes at every vertex of a mesh that covers $\Omega$ is
proposed. The issue of boundary conditions is not treated here.

Let ${\bf a}_i \in \reals^9$ the representation vector at vertex $i$.
The energy function that is considered is pretty standard
\begin{equation}
  E = {1 \over 2} \sum_{ij} \|{A}_i - {A}_j\|_F^2
\end{equation}
where $\sum_{ij} $ is the sum over all edges of the mesh
and $\|\cdot\|_F$ is the Frobenius norm. The energy is minimized in an
explicit fashion. Tensor representations ${\bf a}_i \in \Rset^9$ are
averaged at every vertex of the mesh and subsequently projected back
to $\mbox{SO}(3)$ in the approximate fashion developped above. The
algorith is stopped when the global energy $E$ has decreased by a factor of $10^4$.

We have generated three uniform meshes of a unit sphere with different resolutions. Results are
presented in Figure \ref{fig:sphere}. The iterations were started with
every node assigned to the reference frame aligned with the axis. 
Crosses with values of $\eta$ in the range  $\eta \in [0.3,0.5]$ are 
drawn on the Figures.
Figures show the usual polycube decomposition of the sphere with 
$12$ singular lines made of ``cylinders'' that form an internal topological cube plus 
$8$ singular lines connecting the corners of the topological cube to
the surface. Refining the mesh allows to produce more detailed
representation of the decomposition. Our method is significantly
faster than the ones using spherical hamrmonics thanks to the
efficient projection operator that only requires to compute
eigenvectors and eigenvalues of $3 \times 3$ and $6 \times 6$
symmetric matrices. 
\begin{figure}[h!]
  \begin{tabular}{ccc}
    \includegraphics[width=0.3\textwidth]{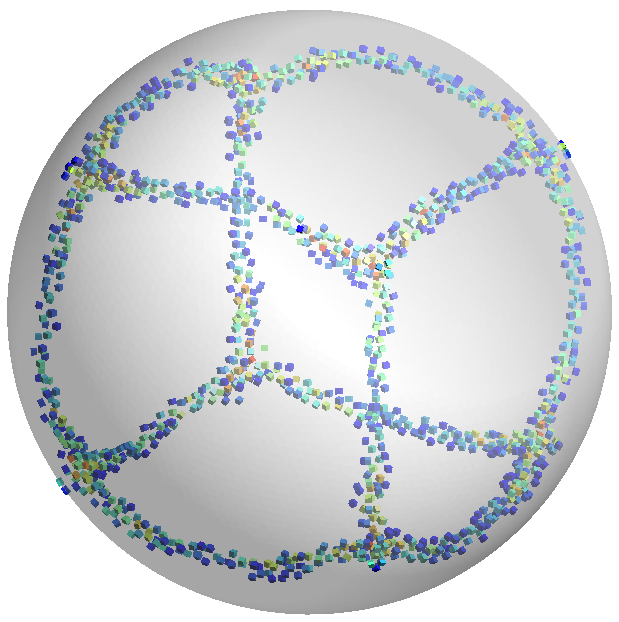}&
    \includegraphics[width=0.3\textwidth]{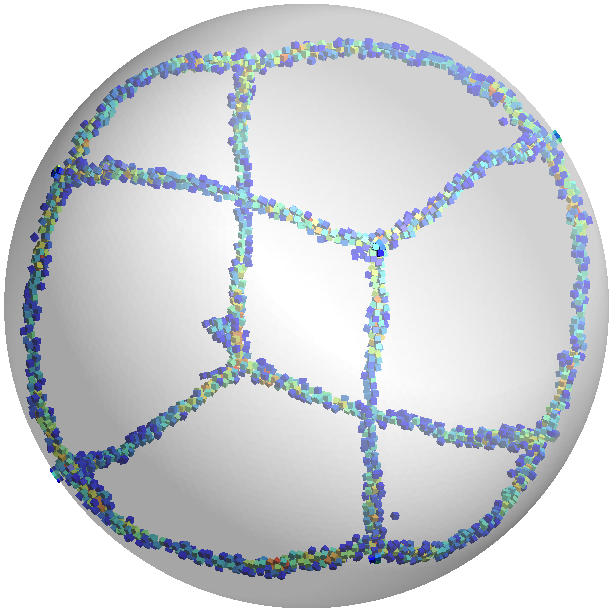}&
    \includegraphics[width=0.3\textwidth]{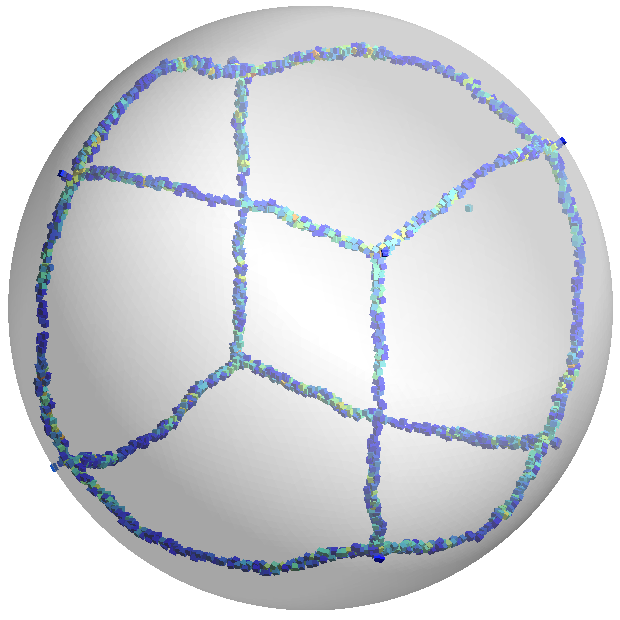}
  \end{tabular}
  \caption{Computation of cross field on a sphere meshed with
    tetrahedra. The three meshes used contain respectively 447,405,
    2,124,801 and 6,128,555 tetrahedra. Resolution time for reducing 
    the residual from $1$ to $10^{-5}$ was respectively $3$ seconds, 
    $34$ seconds and $81$ seconds.\label{fig:sphere}}
\end{figure}

\section{Conclusion}

The method to represent and compute crossfields on 3D domains offer a
lot of advantages.At first, the new formulation is, to our opinion,
way easier to understand geometrically: rotations of tensors, recovery
procedures and projections have a clear geometrical
representation. Then, we have shown that there exist a one-to-one
relationship between our representation and 4th order spherical
harmonics. It should be possible to build a $9 \times 9$ matrix that
allow to change of base.  The $4^{th}$ order tensor representation
used allows to approximate in a very efficient way projections on the
crosses space $\cspace$. The direct consequence is a fast resolution
of the smoothing problem. 

We are aware that this paper is quite theoretical: way more
practical results about this new representation are in our hands:
detection of singularities, boundary conditions, norms... Due to page
limitations, we have deliberatlely made the choice to present basic
results. More practical aspects of that new representationas as well as computations of cross
fields on complex geometries will appear in furthcoming papers.

\begin{acknowledgement}
  This research is supported by the European Research Council (project HEXTREME,
  ERC-2015-AdG-694020) and by the Fond de la Recherche Scientifique de Belgique
  (F.R.S.-FNRS).
\end{acknowledgement}

\bibliographystyle{plain}
\bibliography{ref}   % name your BibTeX data base

\end{document}